\documentclass[11pt,letterpaper]{article}

\usepackage{times}
\usepackage{amsmath,amsfonts}
\usepackage{amssymb}
\usepackage{amsthm}
\usepackage{graphicx}
\usepackage{verbatim}
\usepackage{mathrsfs}
\usepackage{algorithm,algorithmic}
\usepackage{url}

\usepackage{fullpage}

\usepackage[T1]{fontenc}

\newtheorem{lemma}{Lemma}
\newtheorem{thm}{Theorem}
\newtheorem{cor}{Corollary}
\newtheorem{defn}{Definition}
\newtheorem{obs}{Observation}

\newcommand{\E}{\mathrm{E}}

\algsetup{
	linenodelimiter=\mbox{ }
}

\title{LP-Based Approximation Algorithms for Traveling Salesman Path Problems}

\author{ 
Hyung-Chan An
\thanks{
{\tt anhc@cs.cornell.edu}.
Dept. of Computer Science, Cornell University, Ithaca, NY 14853.
Research supported in part by NSF under grants no. CCF-1017688 and CCF-0729102, and the Korea Foundation for Advanced Studies.
Part of this research was conducted while the author was a visiting student at CSAIL, MIT.
}
\and
David B. Shmoys
\thanks{
{\tt shmoys@cs.cornell.edu}.
School of ORIE and Dept. of Computer Science, Cornell University, Ithaca, NY 14853.
Research supported in part by NSF under grants no. CCF-0832782 and CCF-1017688.
Part of this research was conducted while the author was a visiting professor at Sloan School of Management, MIT.
}
}

\date{}

\begin{document}

\maketitle 
\def\thepage {} 
\thispagestyle{empty}

\begin{abstract}
We present a $(\frac{5}{3}-\epsilon)$-approximation algorithm for some constant $\epsilon>0$ for the traveling salesman path problem under the unit-weight graphical metric, and prove an upper bound on the integrality gap of the path-variant Held-Karp relaxation both under this metric and the general metric. Given a complete graph with the metric cost and two designated endpoints in the graph, the traveling salesman path problem is to find a minimum Hamiltonian path between these two endpoints. The best previously known performance guarantee for this problem was $5/3$ and was due to Hoogeveen. We give the first constant upper bound on the integrality gap of the path-variant Held-Karp relaxation, showing it to be at most $5/3$ by providing a new analysis of Hoogeveen's algorithm. This analysis exhibits a well-characterized critical case, and we show that the recent result of Oveis Gharan, Saberi and Singh on the traveling salesman circuit problem under the unit-weight graphical metric can be modified for the path case to complement Hoogeveen's algorithm in the critical case, providing an improved performance guarantee of $(\frac{5}{3}-\epsilon)$. This also proves the matching integrality gap upper bound of $(\frac{5}{3}-\epsilon)$.
\end{abstract}
\newpage
\pagenumbering {arabic}
\normalsize

\section{Introduction}

The traveling salesman problem (TSP) has been widely studied in a variety of settings \cite{HK, C, L, FGM, PR, LLRS, KW, CP, AGMOS, AKS}, but we shall focus on the approximation results for the traveling salesman \emph{path} problem with two prespecified endpoints. In this variant, two vertices in the graph are specified as a part of the input, and we need to find a Hamiltonian path (in contrast to Hamiltonian circuit found in other variants) between these two vertices. We will assume a symmetric metric cost: costs are defined over the complete graph; the costs between two vertices in opposite directions do not differ; and the costs satisfy the triangle inequality. The triangle inequality can be assumed without loss of generality by allowing vertices to be visited more than once; furthermore, it is directly satisfied by many cost functions of interest. This assumption is necessary, in that without the triangle inequality the problem cannot be approximated within any polynomial approximation factor unless $\mathrm{P}=\mathrm{NP}$.

For the traveling salesman circuit problem, a $3/2$-approximation algorithm due to Christofides \cite{C} attains the best performance guarantee known; the matching upper bound on the integrality gap of the Held-Karp relaxation \cite{HK}, a standard LP relaxation of the problem, is obtained by an analysis of this algorithm \cite{W, SW}. In contrast, for the traveling salesman path problem with two prespecified endpoints, the best approximation ratio known is only $5/3$, as achieved by Hoogeveen's algorithm \cite{H}; no upper bound on the integrality gap of the (accordingly modified) Held-Karp relaxation is known. Recently, Oveis Gharan, Saberi and Singh \cite{OSS} studied a special case of the traveling salesman circuit problem where the cost function is a shortest-path metric defined by an unweighted undirected graph, and gave a $(\frac{3}{2}-\epsilon_0)$-approximation algorithm; its analysis constructively proves that the Held-Karp relaxation under this special case has an integrality gap of at most $\frac{3}{2}-\epsilon_0$. In this paper, we first present an LP-based new analysis of Hoogeveen's algorithm that establishes an upper bound of $5/3$ on the integrality gap of the path-variant Held-Karp relaxation under the general metric. We then modify Oveis Gharan, Saberi and Singh's algorithm in order to apply it to the path problem; our new analysis of Hoogeveen's algorithm reveals that the combination of these two algorithms yields a $(\frac{5}{3}-\epsilon)$-approximation algorithm for some constant $\epsilon >0$. This analysis would also constructively prove the matching integrality gap upper bound. We note that recent independent work due to M\"omke and Svensson \cite{MS} yields a better performance guarantee for the same special case of both the circuit- and path-variants, but shows the matching integrality gap upper bound only for the circuit-variant Held-Karp relaxation; it instead shows for the path-variant case that, for any $\epsilon'>0$, there exists $\tau(\epsilon')$ that depends on $\epsilon'$ such that every graph with more than $\tau(\epsilon')$ vertices has an integral optimal solution of value no more than $3-\sqrt{2}+\epsilon'$ times the Held-Karp optimum.

Proposed by Held and Karp \cite{HK} originally for the symmetric traveling salesman \emph{circuit} problem, the Held-Karp relaxation is a linear program relaxation of the traveling salesman polytope. It consists of the degree-2 constraints, and the edge connectivity constraints stating that the graph needs to be 2-edge-connected; the asymmetric Held-Karp relaxation is similarly formulated. These relaxations have proven to be useful in designing approximation algorithms in many settings (see, e.g., \cite{AGMOS, AKS, OSS, MS}). In the LP-based design of an approximation algorithm, one important measure of the strength of a particular LP relaxation is its integrality gap, i.e., the worst-case ratio between the integral and fractional optimal values. There is a significant gap between currently known lower and upper bounds on the integrality gap of the Held-Karp relaxation. For the asymmetric circuit case, the best lower bound known, due to Charikar, Goemans and Karloff \cite{CGK}, is $2$, whereas the upper bound, constructively proven by the best known approximation algorithm of Asadpour, Goemans, M{\k{a}}dry, Oveis Gharan and Saberi \cite{AGMOS}, is $O(\log n/\log\log n)$. For the symmetric circuit case, the best lower bound known is $4/3$, achieved by the family of graphs depicted in Figure~\mbox{\ref{f:1}(a)} under the shortest-path metric \cite{G}; yet, the best upper bound known, also constructively proven \cite{W, SW} based on Christofides' algorithm \cite{C}, is $3/2$. Christofides' algorithm does not directly compute the LP optimum, but the analyses due to Wolsey \cite{W} and Shmoys \& Williamson \cite{SW} establish that the output solution value can be bounded using the LP optimum. 

\begin{figure}
\label{f:1}
\center
\includegraphics[width=.75\textwidth]{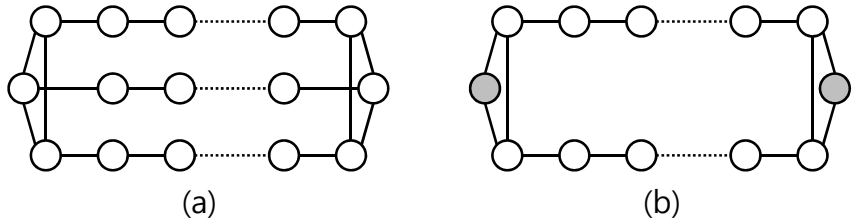}
\caption{Examples establishing the integrality gap lower bounds for the circuit- and path-variant Held-Karp relaxations}
\end{figure}

For the path problem, Hoogeveen \cite{H} presents Christofides-like approximation algorithms for three cases, depending on the number of prespecified endpoints.  When both endpoints are prespecified, Hoogeveen \cite{H} gives a $5/3$-approximation algorithm. However, the analysis compares the output solution value to the optimal (integral) solution, and therefore it is unclear whether the algorithm yields an integrality gap bound for the Held-Karp relaxation analogously formulated for the path problem. We observe that the family of graphs in Figure~\mbox{\ref{f:1}(b)} demonstrates the integrality gap lower bound of $3/2$. Note that this lower bound is strictly greater than the known upper bound on the integrality gap of the circuit-variant Held-Karp relaxation under the shortest-path metric defined by an unweighted undirected graph; this suggests that the lack of a performance guarantee known matching the $3/2$ for all other TSP variants has a true structural cause.

We present in this paper a new analysis of Hoogeveen's algorithm, bounding the output solution value in terms of the LP optimum. In particular, we produce two different bounds that both depend on the cost of the direct edge between the two endpoints but with opposite signs; taking the minimum of these two bounds will establish the integrality gap upper bound of $5/3$.

These two bounds coincide when the edge cost is exactly $1/3$ times the LP optimum; this in fact is the critical case that determines the integrality gap bound proven, and the analysis provides a better ratio on all other cases. Thus, the question whether this critical case can be improved, possibly by using a different algorithm only in the critical case or formulating a third bound that does not coincide with the other two at the critical point, could naturally follow. We will illustrate the former approach in a special case where the cost function is a shortest-path metric defined by an unweighted undirected graph: we will show how Oveis Gharan, Saberi and Singh's algorithm can be modified for the path case and prove that its performance guarantee \emph{near} the critical point is slightly better than Hoogeveen's algorithm. This analysis would also serve as a constructive integrality gap proof for the special case.

Section~\ref{s:pre} of this paper introduces some definitions and notation; Section~\ref{s:5over3} presents the new analysis yielding a constructive integrality gap proof for the general metric. Section~\ref{s:gm} gives the improved approximation algorithm for the special case of the shortest-path metric defined by an unweighted undirected graph, proving the matching integrality gap upper bound.

\section{Preliminaries}\label{s:pre}

In this section, we introduce some definitions and notation to be used throughout this paper.

Let $G=(V,E)$ be the input complete graph with cost function $c:E\to\mathbb{R}_+$. Let $s$ and $t$ be the two prespecificed endpoints; we call the other vertices \emph{internal points}. For $U\subset V$, $\delta(U)$ denotes the set of cut edges: i.e., $\delta(U)=\left\{\{u,v\}\in E : |\{u,v\}\cap U|=1\right\}$. Let $G(U)$ denote the subgraph induced by $U\subset V$.

For $F\subset E$ and $x\in\mathbb{R}^E$, $x(F)$ is a shorthand for $\sum_{f\in F}x_f$; $c(x)$ is for $\sum_{e\in E}x_e c(e)$.

\begin{defn}[Held and Karp \cite{HK}]
\label{d:hkcircuit}
The \textbf{circuit-variant Held-Karp relaxation} is the following:\begin{equation}\label{e:hkcircuit}
\begin{array}{lll}
\textrm{minimize}&c(x)&\\
\textrm{subject to}&x(\delta(S))\geq 2,&\forall S\subsetneq V, S\neq\emptyset;\\
&x(\delta(\{v\})) = 2,&\forall v\in V;\\
&x\geq 0.&
\end{array}
\end{equation}
\end{defn}

\begin{defn}
\label{d:hkpath}
The \textbf{path-variant Held-Karp relaxation} is analogously defined as follows:\begin{equation}\label{e:hkpath1}
\begin{array}{lll}
\textrm{minimize}&c(x)&\\
\textrm{subject to}&x(\delta(S))\geq 1,&\forall S\subsetneq V, |\{s,t\}\cap S| = 1;\\
&x(\delta(S))\geq 2,&\forall S\subsetneq V, |\{s,t\}\cap S| \neq 1, S\neq\emptyset ;\\
&x(\delta(\{s\})) = x(\delta(\{t\})) = 1;&\\
&x(\delta(\{v\})) = 2,&\forall v\in V\setminus\{s,t\};\\
&x\geq 0.&
\end{array}
\end{equation}
\end{defn}

Both linear programs can be solved in polynomial time via the ellipsoid method using a min-cut algorithm to solve the separation problem \cite{GLS}. The following observation gives an equivalent formulation of \eqref{e:hkpath1}.

\begin{obs}
\label{o:hkequiv}
Following is an equivalent formulation of \eqref{e:hkpath1}:\begin{equation}\label{e:hkpath2}
\begin{array}{lll}
\textrm{minimize}&c(x)&\\
\textrm{subject to}&x(E(S))\leq |S|-1,&\forall S\subsetneq V, \{s,t\} \not\subseteq S, S\neq\emptyset;\\
&x(E(S))\leq |S|-2,&\forall S\subsetneq V, \{s,t\} \subseteq S;\\
&x(\delta(\{s\})) = x(\delta(\{t\})) = 1;&\\
&x(\delta(\{v\})) = 2,&\forall v\in V\setminus\{s,t\};\\
&x\geq 0.&
\end{array}
\end{equation}
\end{obs}

\begin{defn}
For $T\subset V$ and $J\subset E$, $J$ is a \textbf{$T$-join} if the set of odd-degree vertices in $G'=(V,J)$ is $T$.
\end{defn}

The following theorem gives a nice polyhedral representation of $T$-joins.

\begin{thm}[Edmonds and Johnson \cite{EJ}]
\label{t:tjoinpr}
Let $P_T(G)$ be the convex hull of the incidence vectors of the $T$-joins on $G=(V,E)$. $P_T(G)+\mathbb{R}_+^E$ can be exactly characterized by
\begin{equation}\label{e:tjoinpr}
\begin{cases}
y(\delta^+(S))\geq 1 ,&\forall S\subset V , |S\cap T|\textnormal{ odd};\\
y\in \mathbb{R}_+^E .&
\end{cases}
\end{equation}
\end{thm}

Following are from Oveis Gharan, Saberi and Singh \cite{OSS}. In these, the parameters $\epsilon_1, \epsilon_2, \gamma, \delta$ and $\rho$ can be chosen as follows: $\epsilon_1 = 1.875\cdot 10^{-12}$, $\epsilon_2 = 5\cdot 10^{-2}$, $\gamma = 10^{-7}$, $\delta=6.25 \cdot 10^{-16}$, $\rho = 1.5\cdot 10^{-24}$, and $n$ denotes $|V|$.

\begin{defn}[Nearly integral edges]\label{d:niedges}
An edge $e$ is \textbf{nearly integral} with respect to $x\in\mathbb{R}^E$ if $x_e\geq 1-\gamma$.
\end{defn}

\begin{defn}\label{d:approxm}
For some constant $\nu\leq\frac{1}{5}$ and $k\geq 2$, a \textbf{maximum entropy distribution over spanning trees with approximate marginal} $x\in\mathbb{R}^E$ is a probability distribution $\mu$ defined by $\lambda\in\mathbb{R}^E$ such that $\mu(T)\propto\prod_{e\in T}\lambda_e$ for every spanning tree $T$ and the marginal probability of every edge $e$ is no greater than $(1+\frac{\nu}{n^k})x_e$.
\end{defn}

\begin{defn}[Good edges]\label{d:goodedges}
A cut is $(1+\delta)$-near-minimum if its weight is at most $(1+\delta)$ times the minimum cut weight. An edge $e$ is \textbf{even} with respect to $F\subset E$ if every $(1+\delta)$-near-minimum cut containing $e$ has even number of edges intersecting with $F$.

For a circuit-variant Held-Karp feasible solution $x^*_{\mathrm{circuit}}$, consider $x^*_{\mathrm{circuit}}$ as the edge weight and let $F$ be a spanning tree sampled from a maximum entropy distribution with approximate marginal $(1-\frac{1}{n}) x^*_{\mathrm{circuit}}$. We say an edge $e$ is \textbf{good} with respect to $x^*_{\mathrm{circuit}}$ if the probability that $e$ is even with respect to $F$ is at least $\rho$.
\end{defn}

\begin{thm}[Structure Theorem]
\label{t:structure}
Let $x^*_{\mathrm{circuit}}$ be a feasible solution to the circuit-variant Held-Karp relaxation, and let $\mu$ be a maximum entropy distribution over spanning trees with approximate marginal $(1-\frac{1}{n}) x^*_{\mathrm{circuit}}$. There exist small constants $\epsilon_1 , \epsilon_2 >0$ such that at least one of the following is true:\begin{itemize}
\item[1.] there exists a set $E^*\subset E$ such that $x(E^*)\geq \epsilon_1 n$ and every edge in $E^*$ is good with respect to $x^*_{\mathrm{circuit}}$;
\item[2.] there exist at least $(1-\epsilon_2)n$ edges that are nearly integral with respect to $x^*_{\mathrm{circuit}}$.
\end{itemize}
\end{thm}

\begin{lemma}
\label{l:case1}
Suppose that Case~1 of Theorem~\ref{t:structure} holds and $F$ is sampled from $\mu$. Let $T$ be the set of odd-degree vertices in $F$, then a minimum $T$-join $J$ satisfies\[
\E[c(J)]\leq c(x^*_{\mathrm{circuit}}) (\frac{1}{2}-\frac{\epsilon_1\delta\rho}{4(1+\delta)})
.\]
\end{lemma}

\section{An LP-based new analysis of Hoogeveen's $5/3$-approximation algorithm}\label{s:5over3}

In this section, we present a new analysis of Hoogeveen's algorithm for the traveling salesman path problem with two prespecified endpoints. This analysis compares the output solution value to the LP optimum of the path-variant Held-Karp relaxation, although the LP optimum is never computed by the algorithm.

Hoogeveen's algorithm, when both endpoints are prespecified, is shown in Algorithm~\ref{a:hoogeveen}.

\begin{algorithm}[H]
\caption{Hoogeveen's algorithm \cite{H}}
\label{a:hoogeveen}
\begin{algorithmic}[1]
	\REQUIRE Complete graph $G=(V,E)$ with cost function $c:E\to\mathbb{R}_+$; endpoints $s,t\in V$.
	\ENSURE Hamiltonian path between $s$ and $t$
	\STATE Find a minimum spanning tree $H$ of $G$.
	\STATE Let $T\subset V$ be the set of the vertices with the `wrong' parity of degree in $H$:\\i.e., $T$ is the set of odd-degree internal points and even-degree endpoints in $H$.
	\STATE Find a minimum perfect matching $M$ on $G(T)$.
	\STATE Shortcut an Eulerian path of the multigraph $H\cup M$ to obtain a Hamiltonian path; output it.
\end{algorithmic}
\end{algorithm}

Let $x^*\in\mathbb{R}^E$ be the LP optimum of the path-variant Held-Karp relaxation, and $\mathsf{ALG}$ be the output of the Algorithm.

\begin{lemma}
\label{l:h1}
$c(H)\leq c(x^*)$.
\end{lemma}
\begin{proof}
As can be seen from \eqref{e:hkpath2}, the path-variant Held-Karp polytope is contained in the spanning tree polytope. The lemma follows from this observation, since $H$ is a minimum spanning tree.
\end{proof}

Lemmas~\ref{l:h2} and \ref{l:h3} give two different bounds on the cost of $M$.

\begin{lemma}
\label{l:h2}
$c(M)\leq \frac{1}{2}\left\{c(x^*)+c(s,t)\right\}$.
\end{lemma}
\begin{proof}
Let $x^*_{\mathrm{circuit}} := x^*+\mathbf{e}_{(s,t)}$: i.e., $x^*_{\mathrm{circuit}}$ is obtained by `adding' the edge $(s,t)$ to $x^*$. Then $x^*_{\mathrm{circuit}}$ is a feasible solution to the circuit-variant Held-Karp relaxation (see \eqref{e:hkcircuit} and \eqref{e:hkpath1}). Let $\mathrm{HK_{circuit}}$ be the optimal value of the circuit-variant Held-Karp relaxation and we have\begin{eqnarray*}
c(M)&\leq& \frac{1}{2}\mathrm{HK_{circuit}}\\
& \leq& \frac{1}{2} c(x^*_{\mathrm{circuit}})\\
& =& \frac{1}{2} \left\{ c(x^*)+c(s,t) \right\}
,\end{eqnarray*}where the first inequality follows from \cite{W, SW}.
\end{proof}

\mbox{}

\begin{lemma}
\label{l:h3}
$c(M)\leq c(x^*)-c(s,t)$.
\end{lemma}
\begin{proof}
Let $P_{st}^H$ be the path between $s$ and $t$ on $H$. Consider an edge set $M':=H\setminus P_{st}^H$. Note that $M'$ is a $T$-join: $v\in V$ has even degree in $P_{st}^H$ if and only if $v$ is internal; thus, $v$ has even degree in the multigraph $H\cup M'=(H\cup H)\setminus P_{st}^H$ if and only if $v$ is an internal point, and this shows that $v$ has odd degree in $M'$ if and only if $v\in T$.

We have\begin{eqnarray*}
c(M) &\leq& c(M')\\
&=& c(H)-c(P_{st}^H)\\
&\leq& c(x^*) -c(s,t)
.\end{eqnarray*}The last inequality follows from Lemma~\ref{l:h1} and the triangle inequality.
\end{proof}

\begin{thm}
\label{t:general}
$c(\mathsf{ALG})\leq\frac{5}{3}c(x^*)$; therefore, Algorithm~\ref{a:hoogeveen} is a $5/3$-approximation algorithm, and the integrality gap of the path-variant Held-Karp relaxation is at most $5/3$.
\end{thm}
\begin{proof}
We have
\begin{eqnarray}
c(\mathsf{ALG}) &\leq& c(H)+c(M)\nonumber\\
&\leq& c(x^*)+\min\left [\frac{1}{2}\left\{c(x^*)+c(s,t)\right\} , c(x^*)-c(s,t) \right ]\nonumber\\
&=& \frac{5}{3}c(x^*)+\min\left [\frac{1}{2}\left\{-\frac{1}{3}c(x^*)+c(s,t)\right\} , \frac{1}{3}c(x^*)-c(s,t) \right ]\nonumber\\
&\leq& \frac{5}{3}c(x^*)\label{e:crit}
,\end{eqnarray}where the second inequality follows from Lemmas~\ref{l:h1}, \ref{l:h2} and \ref{l:h3}.
\end{proof}

\paragraph{Open Question.} We can observe that the equality of \eqref{e:crit} is achieved when $c(s,t)=\frac{1}{3}c(x^*)$, and this is the critical case of this analysis that determines the integrality gap bound proven. Hence, if we can improve the performance guarantee near this critical case, such an improvement would lead to a better guarantee, and possibly a tighter integrality gap bound. For example, even if an algorithm does not outperform Hoogeveen's algorithm in general, if it yields a better ratio near the critical case, i.e., when $c(s,t)$ is close to $\frac{1}{3}c(x^*)$, this algorithm can be combined with Hoogeveen's to provide a better overall performance guarantee. Another possibility could be to devise a third upper bound on $c(M)$: if such a bound is better than the other two near the critical case, this would lead to a tighter integrality gap bound. It remains an open question whether the analysis for this critical case can be so improved.

\section{Unit-weight graphical metrics}\label{s:gm}

We illustrate in this section how the critical case discussed in Section~\ref{s:5over3} can be improved by introducing another algorithm that bears a better performance guarantee near the critical case. In particular, we show how the present results combine with the recent work due to Oveis Gharan, Saberi and Singh \cite{OSS} on a special case of the traveling salesman circuit problem, to yield an improved constructive integrality gap proof of the path-variant Held-Karp relaxation under a similar special case.

Oveis Gharan et al. consider the case where the metric is a shortest path metric defined by an unweighted undirected graph, and give a $\left(\frac{3}{2}-\epsilon_0 \right)$-approximation algorithm for some constant $\epsilon_0>0$. We will consider the traveling salesman path problem with two specified endpoints under the same class of metric, and show how to modify Oveis Gharan, Saberi and Singh's algorithm for the path case and that, when $c(s,t)$ is close to $\frac{1}{3}c(x^*)$, this modified algorithm carries a performance guarantee that is slightly better than $5/3$.

\begin{algorithm}[b!]
\caption{Algorithm for the shortest-path metric defined by an unweighted undirected graph}
\label{a:sc}
\begin{algorithmic}[1]
	\REQUIRE Complete graph $G=(V,E)$ with cost function $c:E\to\mathbb{R}_+$; endpoints $s,t\in V$.
	\ENSURE Hamiltonian path between $s$ and $t$

	\STATE $x^*\gets$optimal solution to the path-variant Held-Karp relaxation
	\IF {$c(s,t)=(\frac{1}{3}+\alpha)c(x^*)$ for $\alpha\in[-\sigma_l,\sigma_u]$}
		\IF[Case A1]{at least $(1-\epsilon'_2)(n-1)$ edges are nearly integral with respect to $x^*$}
			\STATE Find a minimum spanning subgraph $F'$ containing all the nearly integral edges
			\STATE Find a minimum spanning tree $F$ of $F'$
			\STATE Let $T$ be the set of odd-degree internal points and even-degree endpoints in $F$
			\STATE Compute a minimum $T$-join $J$; $L \gets F\cup J$
		\ELSE[Case A2]
			\STATE $x^*_{\mathrm{circuit}} := x^*+\mathbf{\mathbf{e}}_{(s,t)}$
			\STATE Sample spanning tree $F$ from max-entropy distrib. with approx. marginal $(1-\frac{1}{n}) x^*_{\mathrm{circuit}}$\hspace*{-5em}
			\STATE Let $T$ be the set of odd-degree vertices in $F$
			\STATE Compute a minimum $T$-join $J$; $L_0 \gets F\cup J$
			\STATE \textbf{if} $(s,t)\in L_0$ \textbf{then} $L\gets L_0 \setminus\{(s,t)\}$ \textbf{else} $L\gets L_0 \cup\{(s,t)\}$ \textbf{end if}
		\ENDIF
	\ELSE[Case B, Hoogeveen's algorithm]
		\STATE Find a minimum spanning tree $H$ of $G$.
		\STATE Let $T$ be the set of odd-degree internal points and even-degree endpoints in $H$
		\STATE Compute a minimum $T$-join $J$; $L \gets F\cup J$
	\ENDIF
	\STATE Shortcut an Eulerian path of the multigraph $L$ to obtain a Hamiltonian path; output it.
\end{algorithmic}
\end{algorithm}

Algorithm~\ref{a:sc} gives the entire algorithm for the traveling salesman path problem under the special case. It first computes the LP optimum $x^*$. If $c(s,t)$ is close to $\frac{1}{3}c(x^*)$, we run a modified version of Oveis Gharan, Saberi and Singh's algorithm (Cases~A1 and A2); otherwise, we invoke Hoogeveen's algorithm (Case~B). Let $\mathsf{ALG}$ be the output of the algorithm. Parameters $\sigma_l,\sigma_u$ and $\epsilon'_2$ are to be chosen later.

First we show that we can have a Structure Theorem analogous to Theorem~\ref{t:structure} by adjusting $\epsilon_2$ and replacing $n$ with $(n-1)$ in Case~2. The following corollary states that either there are good edges of significant weight with respect to $x^*_{\mathrm{circuit}}$ or there are many nearly integral edges with respect to $x^*$.
\begin{cor}
\label{c:pathstructure}
Let $x^*$ be a feasible solution to the path-variant Held-Karp relaxation and  $x^*_{\mathrm{circuit}} := x^*+\mathbf{\mathbf{e}}_{(s,t)}$. Let $\mu$ be a maximum entropy distribution over spanning trees with approximate marginal $(1-\frac{1}{n}) x^*_{\mathrm{circuit}}$. There exist small constants $\epsilon_1 , \epsilon'_2 >0$ such that at least one of the following is true:\begin{itemize}
\item[1.] there exists a set $E^*\subset E$ such that $x(E^*)\geq \epsilon_1 n$ and every edge in $E^*$ is good with respect to $x^*_{\mathrm{circuit}}$;
\item[2.] there exist at least $(1-\epsilon'_2)(n-1)$ edges that are nearly integral with respect to $x^*$.
\end{itemize}
\end{cor}
\begin{proof}
By Theorem~\ref{t:structure}, at least one of the two cases of Theorem~\ref{t:structure} holds. Case~1 of Theorem~\ref{t:structure} and Case~1 of this corollary are identical, so consider when Case~2 of Theorem~\ref{t:structure} holds.

Recall that $\epsilon_2$ was chosen as $5\cdot 10^{-2}$; we choose $\epsilon'_2=6\cdot 10^{-2}$.

Suppose $n\leq 19$. $x^*_{\mathrm{circuit}}$ has at least $(1-\epsilon_2)n$ nearly integral edges; thus, $x^*$ has at least $\lceil (1-\epsilon_2)n\rceil - 1 = n-1\geq (1-\epsilon'_2)(n-1)$ nearly integral edges.

Suppose $n\geq 20$. $x^*$ has at least\begin{eqnarray*}
(1-\epsilon_2)n-1&=&(1-\epsilon_2)(n-1)-\epsilon_2\\
&\geq& (1-\frac{20}{19}\epsilon_2)(n-1)\\
&\geq& (1-\epsilon'_2)(n-1)
\end{eqnarray*}nearly integral edges.
\end{proof}

\begin{lemma}
\label{l:fca1}
In Case~A1, $c(\mathsf{ALG})\leq (\frac{5}{3}-C_{A1})c(x^*)$ for some $c_{A1}>0$.
\end{lemma}
\begin{proof}
The following proof is adapted from \cite{OSS} and modified for the path case.

Let $S'$ be the set of nearly integral edges. Since the metric is defined by an unweighted connected graph, $c(F')=c(S')+|F'\setminus S'|\leq\frac{c(x(S'))}{1-\gamma}+|F'\setminus S'|$. From $\gamma<\frac{1}{3}$, we know that $S'$ is a union of disjoint cycles and paths and the lengths of cycles are at least $\frac{1}{\gamma}$. Thus, $|F\cap S'|\geq (n-1)(1-\epsilon'_2)(1-\gamma)$ and $|F\setminus S'|\leq (n-1)(\epsilon'_2+\gamma)\leq c(x^*)(\epsilon'_2+\gamma)$. Let $S=S'\cap F$.

We construct a feasible solution $y$ to \eqref{e:tjoinpr} as follows. Note that Theorem~\ref{t:tjoinpr} implies $c(J)\leq c(y)$.\[
y_e = \begin{cases}
1&\textrm{if }e\in F\setminus S\\
x_e^*&\textrm{if }e\in E\setminus F\\
\frac{x_e^*}{2(1-\gamma)}&\textrm{if }e\in S
\end{cases}
\]Let $(U,\bar U)$ be any cut that has an odd number of vertices in $T$ on one side. If there exists an edge $e\in (F\setminus S)\cap \delta(U)$, then $y(\delta(U))\geq y_e = 1$. So suppose from now on that $\delta(U)\cap F\subset S$. Then $\delta(U)\cap S=\delta(U)\cap F$.

If $s$ and $t$ lie on the same side of the cut, $U$ contains odd number of odd-degree vertices, and thus $|\delta(U)\cap F|$ is odd. We have $x^*(\delta(U))\geq 2$ from the Held-Karp formulation and thus\[
\begin{cases}
y(\delta(U))\geq x^*(\delta(U)\setminus F)\geq 1&\textrm{if }|\delta(U)\cap F|=1\\
y(\delta(U))\geq y(\delta(U)\cap S)\geq 3\frac{1-\gamma}{2(1-\gamma)} > 1&\textrm{if }|\delta(U)\cap S|\geq 3
.\end{cases}
\]

If $(U,\bar U)$ separates $s$ and $t$, then $U$ contains even number of odd-degree vertices, and thus $|\delta(U)\cap F|$ is even. We have $(\delta(U)\cap F)\neq\emptyset$ since $F$ is connected and\[
y(\delta(U))\geq y(\delta(U)\cap S)\geq 2\frac{1-\gamma}{2(1-\gamma)}=1
.\]

Thus $y$ is a feasible solution to $\mathrm{LP_{T-join}}$. Now,\begin{eqnarray*}
c(\mathsf{ALG})&\leq&c(F)+c(y)\\
&\leq&\frac{c(x^*(S))}{1-\gamma}+c(F\setminus S)+c(F\setminus S)+c(x^*(E\setminus F))+\frac{c(x^*(S))}{2(1-\gamma)}\\
&\leq&\frac{3c(x^*(S))}{2(1-\gamma)}+2c(x^*)(\epsilon'_2+\gamma)+c(x^*(E\setminus S))\\
&\leq&c(x^*)(\frac{3}{2(1-\gamma)}+2\epsilon'_2+2\gamma)\\
&\leq&c(x^*)(\frac{5}{3}-C_{A1})
\end{eqnarray*}for some $C_{A1}>0$. For example, we can choose $c_{A1}=4\cdot 10^{-2}$.
\end{proof}

\begin{lemma}
\label{l:fca2}
In Case~A2, $\E[c(\mathsf{ALG})]\leq (\frac{5}{3}-C_{A2})c(x^*)$ for some $C_{A2}>0$.
\end{lemma}
\begin{proof}
First we have\begin{eqnarray*}
\E[c(F)] &\leq& c\left((1+\frac{\nu}{n^k})(1-\frac{1}{n})x^*_{\mathrm{circuit}}\right)\\
&\leq& (1+\frac{1}{5n^2})(1-\frac{1}{n})(\frac{4}{3}+\alpha)c(x^*)\\
&\leq& (1-\frac{4}{5n})(\frac{4}{3}+\alpha)c(x^*)
.\end{eqnarray*}From Lemma~\ref{l:case1},\[
\E[c(J)]\leq(\frac{4}{3}+\alpha)c(x^*)(\frac{1}{2}-\frac{\epsilon_1\delta\rho}{4(1+\delta)})
.\]

We have\begin{eqnarray*}
\Pr[(s,t)\in L_0] &\geq& \Pr[(s,t)\in F]\\
&=& n-1-\E[|F\setminus(s,t)|]\\
&\geq& n-1-(n-2+\frac{1}{n})(1+\frac{\nu}{n^k})\\
&\geq& n-1-(n-2+\frac{1}{n})(1+\frac{1}{5n^2})\\
&\geq& 1-\frac{7}{5n}
\end{eqnarray*}and hence\begin{eqnarray*}
\E[c(\mathsf{ALG})]&\leq&\E[c(F)]+\E[c(J)]-(1-\frac{7}{5n})c(s,t)+\frac{7}{5n}c(s,t)\\
&\leq&c(x^*)\left\{(1-\frac{4}{5n})(\frac{4}{3}+\alpha)+(\frac{4}{3}+\alpha)(\frac{1}{2}-\frac{\epsilon_1\delta\rho}{4(1+\delta)})-(1-\frac{7}{5n})(\frac{1}{3}+\alpha)+\frac{7}{5n}(\frac{1}{3}+\alpha)\right\}\\
&=&c(x^*)\left\{(\frac{5}{3}-\frac{\epsilon_1\delta\rho}{3(1+\delta)})+\alpha(\frac{1}{2}-\frac{\epsilon_1\delta\rho}{4(1+\delta)}) -\frac{1}{n} (\frac{2}{15}-2\alpha) \right\}\\
&\leq&c(x^*)(\frac{5}{3}-C_{A2})
\end{eqnarray*}for some $C_{A2}>0$ by choosing sufficiently small $\sigma_l,\sigma_u>0$. For example, we can choose $\sigma_l=7.8\cdot 10^{-52}$, $\sigma_u=3.9\cdot 10^{-52}$ and $C_{A2}=3.9\cdot 10^{-52}$.
\end{proof}

\begin{lemma}
\label{l:fcb}
In Case~B, $c(\mathsf{ALG})\leq (\frac{5}{3}-C_B)c(x^*)$ for some $C_B>0$.
\end{lemma}
\begin{proof}
Suppose that $c(s,t)<(\frac{1}{3}-\sigma_l)c(x^*)$. From Lemmas~\ref{l:h1} and \ref{l:h2}, it follows that\begin{eqnarray*}
c(\mathsf{ALG})&\leq& c(F)+c(J)\\
&<& c(x^*)+\frac{1}{2} \left\{c(x^*)+ (\frac{1}{3}-\sigma_l)c(x^*) \right\}\\
&=& \left(\frac{5}{3}-\frac{\sigma_l}{2}\right) c(x^*)
.\end{eqnarray*}

Suppose $c(s,t)>(\frac{1}{3}+\sigma_u)c(x^*)$. From Lemmas~\ref{l:h1} and \ref{l:h3},\begin{eqnarray*}
c(\mathsf{ALG})&\leq& c(F)+c(J)\\
&<&c(x^*)+\left\{c(x^*) - (\frac{1}{3}+\sigma_u) c(x^*)\right\}\\
&=& \left(\frac{5}{3}-\sigma_u\right) c(x^*)
.\end{eqnarray*}

Now choose $C_B:=\min(\frac{\sigma_l}{2},\sigma_u)$.
\end{proof}

Lemmas~\ref{l:fca1}, \ref{l:fca2} and \ref{l:fcb} yield the following theorem.

\begin{thm}
\label{t:sc}
For some $\epsilon >0$, Algorithm~\ref{a:sc} is a $(\frac{5}{3}-\epsilon)$-approximation algorithm for the traveling salesman path problem with two prespecified endpoints under the shortest-path metric defined by an unweighted undirected graph.
\end{thm}
\begin{proof}
In Cases~A1 and B, the multigraph $L$ is the union of a spanning tree and a $T$-join where $T$ is the set of the vertices with the wrong parity of degree. Thus, $L$ has an Eulerian path between the two endpoints.

In Case~A2, $L_0$ is Eulerian and hence 2-edge-connected; $L\supset L_0 \setminus\{(s,t)\}$ is therefore connected and $L$ has an Eulerian path between the two endpoints.

By choosing $\epsilon=\min \{C_{A1}, C_{A2}, C_B\}$, $\epsilon =3.9\cdot 10^{-52}$ for example, we have $\E[c(\mathsf{ALG})]\leq (\frac{5}{3}-\epsilon )c(x^*)$ from Lemmas~\ref{l:fca1}, \ref{l:fca2} and \ref{l:fcb}. Thus, Algorithm~\ref{a:sc} is a $(\frac{5}{3}-\epsilon)$-approximation algorithm.
\end{proof}

Note that this proof also establishes the matching integrality gap upper bound.

\paragraph{Acknowledgments.} The authors thank Bobby Kleinberg for very helpful discussions.

\bibliography{path}

\begin{thebibliography}{10}

\bibitem{AKS}
H.-C. An, R.~D. Kleinberg, and D.~B. Shmoys.
\newblock Approximation algorithms for the bottleneck asymmetric traveling
  salesman problem.
\newblock In {\em APPROX-RANDOM}, pages 1--11, 2010.

\bibitem{AGMOS}
A.~Asadpour, M.~X. Goemans, A.~M{\k{a}}dry, S.~{Oveis Gharan}, and A.~Saberi.
\newblock An {$O(\log n/\log \log n)$}-approximation algorithm for the
  asymmetric traveling salesman problem.
\newblock In {\em SODA '10: Proceedings of the 21st Annual ACM-SIAM Symposium
  on Discrete Algorithms}, pages 379--389, 2010.

\bibitem{CGK}
M.~Charikar, M.~X. Goemans, and H.~Karloff.
\newblock {On the integrality ratio for asymmetric TSP}.
\newblock In {\em Proceedings of the 45th Annual IEEE Symposium on Foundations
  of Computer Science}, pages 101--107. IEEE Computer Society, 2004.

\bibitem{CP}
C.~Chekuri and M.~P{\'a}l.
\newblock An {$O(\log n)$} approximation ratio for the asymmetric traveling
  salesman path problem.
\newblock {\em Theory of Computing}, 3(1):197--209, 2007.

\bibitem{C}
N.~Christofides.
\newblock Worst-case analysis of a new heuristic for the travelling salesman
  problem.
\newblock Technical Report 388, Graduate School of Industrial Administration,
  CMU, 1976.

\bibitem{EJ}
J.~Edmonds and E.~Johnson.
\newblock Matching: A well-solved class of integer linear programs.
\newblock In M.~Junger, G.~Reinelt, and G.~Rinaldi, editors, {\em Combinatorial
  Optimization --- Eureka, You Shrink!}, volume 2570 of {\em Lecture Notes in
  Computer Science}, pages 27--30. Springer Berlin / Heidelberg, 2003.

\bibitem{FGM}
A.~M. Frieze, G.~Galbiati, and F.~Maffioli.
\newblock On the worst-case performance of some algorithms for the asymmetric
  traveling salesman problem.
\newblock {\em Networks}, 12:23--39, 1982.

\bibitem{G}
M.~X. Goemans.
\newblock {Worst-case comparison of valid inequalities for the TSP}.
\newblock {\em Mathematical Programming}, 69:335--349, 1995.

\bibitem{GLS}
M.~Gr{\"o}tschel, L.~Lov{\'a}sz, and A.~Schrijver.
\newblock The ellipsoid method and its consequences in combinatorial
  optimization.
\newblock {\em Combinatorica}, 1(2):169--197, 1981.

\bibitem{HK}
M.~Held and R.~M. Karp.
\newblock The traveling-salesman problem and minimum spanning trees.
\newblock {\em Operations Research}, 18(6):1138--1162, 1970.

\bibitem{H}
J.~A. Hoogeveen.
\newblock {Analysis of Christofides' heuristic: some paths are more difficult
  than cycles}.
\newblock {\em Operations Research Letters}, 10(5):291 -- 295, 1991.

\bibitem{KW}
J.~Kleinberg and D.~P. Williamson.
\newblock Unpublished manuscript, 1998.
\newblock Also see \url{http://legacy.orie.cornell.edu/~dpw/cornell.ps}, pages
  124--126.

\bibitem{L}
H.~T. Lau.
\newblock {Finding EPS-Graphs}.
\newblock {\em Monatshefte f\"ur Mathematik}, 92(1):37--40, Mar 1981.

\bibitem{LLRS}
E.~L. Lawler, J.~K. Lenstra, A.~H.~G. {Rinnooy Kan}, and D.~B. Shmoys, editors.
\newblock {\em The Traveling Salesman Problem: A Guided Tour of Combinatorial
  Optimization}.
\newblock Wiley, 1985.

\bibitem{MS}
T.~M{\"o}mke and O.~Svensson.
\newblock {Approximating graphic TSP by matchings}.
\newblock {\em CoRR}, abs/1104.3090, 2011.
\newblock \url{http://arxiv.org/abs/1104.3090}.

\bibitem{OSS}
S.~{Oveis Gharan}, A.~Saberi, and M.~Singh.
\newblock A randomized rounding approach to the traveling salesman problem.
\newblock \url{http://www.stanford.edu/~shayan/Publications_files/tsp.pdf}.

\bibitem{PR}
R.~G. Parker and R.~L. Rardin.
\newblock Guaranteed performance heuristics for the bottleneck traveling
  salesman problem.
\newblock {\em Operations Research Letters}, 2(6):269 -- 272, 1984.

\bibitem{SW}
D.~B. Shmoys and D.~P. Williamson.
\newblock {Analyzing the Held-Karp TSP bound: a monotonicity property with
  application}.
\newblock {\em Information Processing Letters}, 35(6):281--285, 1990.

\bibitem{W}
L.~A. Wolsey.
\newblock Heuristic analysis, linear programming and branch and bound.
\newblock In {\em Combinatorial Optimization II}, volume~13 of {\em
  Mathematical Programming Studies}, pages 121--134. Springer Berlin
  Heidelberg, 1980.

\end{thebibliography}

\end{document}